\documentclass[11pt]{article}

\usepackage[T1]{fontenc}
\usepackage[utf8]{inputenc}
\usepackage{algorithm}
\usepackage{algpseudocode}
\usepackage{amsmath}
\usepackage{amssymb}
\usepackage{amsthm}
\usepackage{booktabs}
\usepackage{caption}
\usepackage{enumitem}
\usepackage{float}
\usepackage[margin=1in]{geometry}
\usepackage{graphicx}
\usepackage{microtype}
\usepackage[numbers,sort&compress]{natbib}
\usepackage{tikz}
\usepackage{xcolor}
\usepackage[colorlinks=true,linkcolor=blue,citecolor=blue,urlcolor=blue]{hyperref}
\usetikzlibrary{arrows.meta,positioning,shapes.geometric}

\setlist{itemsep=0.15em,topsep=0.35em}

\newtheorem{theorem}{Theorem}[section]
\newtheorem{lemma}[theorem]{Lemma}
\newtheorem{proposition}[theorem]{Proposition}
\newtheorem{definition}[theorem]{Definition}

\newcommand{\BDTS}{BDTS}

\newcommand{\History}{\mathcal{H}}

\newcommand{\Active}{\mathsf{active}}
\newcommand{\Closed}{\mathsf{closed}}
\newcommand{\tokens}{\mathrm{tok}}
\newcommand{\bytes}{\mathrm{bytes}}
\newcommand{\children}{\mathrm{children}}
\newcommand{\desc}{\mathrm{desc}}

\title{Budgeted Dynamic Trace Structures for Token-Efficient Sequential Computation}

\author{%
Faruk Alpay\\
Department of Computer Engineering\\
Bahçeşehir University, Istanbul, Turkey\\
\texttt{faruk.alpay@bahcesehir.edu.tr}
\and
Levent Sarioğlu\\
Department of Computer Engineering\\
Bahçeşehir University, Istanbul, Turkey\\
\texttt{levent.sarioglu@bahcesehir.edu.tr}
}
\date{}

\begin{document}

\maketitle

\begin{abstract}
Sequential computation increasingly produces long traces containing nested
branches, status transitions, textual payloads, and compact summaries of
earlier execution.  This paper introduces \emph{budgeted dynamic trace
structures} (\BDTS), a data-structural framework for maintaining rooted
trace graphs and append-only histories under an explicit byte or token
budget.  \BDTS{} combines status-filtered reachability, cursor pagination,
soft-capped recency logs, reference-counted observation keys, delta
overlays, bounded cost caches, and summary-plus-suffix compaction.  We give
formal invariants, asymptotic bounds, and an ancillary Rust implementation
with reproducible benchmarks.  Across synthetic traces with 10,000--40,000
vertices, the prototype builds graphs in 0.58--2.72 ms, enumerates all
descendants in 0.24--1.42 ms, and compacts histories of 350k--2.71M
approximate tokens to 1,048--4,120 approximate tokens; tokenizer and
forward measurements with three public model targets reduce 3,359--3,360
trace tokens to 432--433 tokens.
\end{abstract}

\section{Introduction}

Long-running sequential computations leave traces.  A trace may contain
small state transitions, large textual payloads, references to external
objects, branch points, and later repairs to earlier branches.  In many
systems the trace is not merely an audit log.  It is an active data
structure: later computation asks for the current descendants of a node,
filters a branch by status, scans recent events, loads a page of older
items, or transmits a compact view to a downstream procedure.  If the
available representation budget is fixed, the trace must also be shortened
without destroying the operational boundary between old summarized material
and new material still worth preserving verbatim.

The classical vocabulary of data structures gives strong tools for parts
of this problem.  Depth-first and breadth-first traversal give linear-time
reachability in static graphs \cite{tarjan1972dfs,cormen2009introduction}.
Dynamic graph algorithms analyze update/query tradeoffs under edge changes
\cite{even1981online,henzinger1999randomized,demetrescu2004dynamic,sankowski2004dynamic,roditty2008dynamic}.
Succinct data structures reduce static tree and sequence representations
\cite{jacobson1989succinct,raman2002succinct,navarro2016compact}.  Caching
theory explains bounded auxiliary storage \cite{sleator1985amortized,megiddo2003arc}.
Streaming algorithms and sliding-window summaries address online retention
\cite{muthukrishnan2005streams,datar2002sliding,cormode2005countmin}.  Data
compression and subword tokenization describe how symbolic material can be
shortened or measured \cite{ziv1977universal,ziv1978compression,welch1984technique,gage1994bpe,sennrich2016bpe,kudo2018sentencepiece}.
Yet the combined trace problem has a different shape.  The representation
budget is part of the interface, not merely a storage optimization; the
graph and the history must remain coordinated; and compaction must produce a
valid replacement trace, not only a compressed byte stream.

This paper formulates the combined problem as \BDTS{}.  The core object is a
rooted directed trace graph, in which each non-root vertex has one current
parent and each edge has a state from a small finite set.  The graph is
paired with an append-only history of trace items.  A budget policy maps
either bytes or tokenizer units to a common accounting interface.  A
compaction operation selects a recent suffix subject to that budget,
truncates the boundary item if needed, and prepends a summary item that
records the abstract effect of the discarded prefix.  The design is simple
enough to serve as an ancillary artifact, but precise enough to support
complexity analysis and reproducible measurements.

The phrase ``token-efficient'' is used here in a strictly data-structural
sense.  Tokens are a representation-cost unit induced by a tokenizer, just
as bytes are a representation-cost unit induced by an encoding.  The paper
does not depend on any semantic claim about the content of the trace.  It
asks a concrete question: how should a dynamic trace be maintained when the
system repeatedly needs a bounded, valid, recent, and queryable view?

\subsection{Contributions}

The paper makes four contributions.

\begin{enumerate}[leftmargin=2em]
\item We define \BDTS, a unified data-structure interface for status-filtered
rooted reachability, append-only trace histories, and bounded replacement
views.
\item We give algorithms for edge upsert, edge-state update, state-filtered
child listing, state-filtered descendant traversal, history append, and
budgeted compaction.
\item We prove bounds for the abstract structure and identify the constants
that matter in practical token-budget settings.
\item We provide an ancillary Rust crate and measured benchmark results,
including a concrete tokenizer and forward measurement using the
\texttt{distilbert/distilgpt2} model card target \cite{hfDistilgpt2,sanh2019distilbert}.
\end{enumerate}

\subsection{Why a New Formulation Is Useful}

A trace with branches can be stored as a flat log, but flat storage makes
descendant queries expensive or forces a later index.  A trace graph can
answer reachability queries, but graph storage alone does not describe how
to fit a long history into a fixed representation budget.  A compressor can
reduce bytes, but a compressor does not usually preserve a recent suffix as
explicit trace items.  A sliding window preserves a suffix, but it does not
carry a summary of the discarded prefix.  \BDTS{} places these operations in
one interface, so that their invariants can be checked together.

This formulation also clarifies a computational-complexity boundary.  Fully
dynamic reachability in directed graphs is a difficult problem under general
insertions and deletions.  \BDTS{} deliberately uses a restricted update model:
each child has one current parent, state transitions do not remove the edge
record, and descendant queries are rooted and state-filtered.  This
restriction yields an inexpensive structure with predictable behavior, while
still covering the common trace pattern of a computation that branches from
earlier states and later marks subtraces as active or closed.

\section{Model}

\subsection{Trace Graphs}

Let $V$ be a finite set of trace identifiers with distinguished root $r$.
Let $\Sigma$ be a finite set of edge states.  In the prototype and
experiments we use $\Sigma=\{\Active,\Closed\}$, but the algorithms extend
to any constant-size state set.  A trace graph is a directed graph
$G=(V,E)$ in which each edge is a triple $(u,v,\sigma)$ with parent $u$,
child $v$, and state $\sigma\in\Sigma$.

\begin{definition}[Current-parent invariant]
A trace graph satisfies the current-parent invariant if each non-root vertex
$v$ appears as the child of at most one edge.  Equivalently, there is a
partial function $p:V\setminus\{r\}\to V$ such that
$(p(v),v,\sigma)\in E$ for one state $\sigma$ whenever $v$ has a parent.
\end{definition}

The current-parent invariant does not prevent a vertex from acquiring a new
parent over time.  An upsert operation may move $v$ from an old parent to a
new parent.  The invariant says that after the operation completes the graph
contains one current parent for $v$.  This is the natural dynamic analogue
of an arborescent trace: a vertex records the computation state from which
it currently branches.

For a state predicate $P:\Sigma\to\{\mathsf{true},\mathsf{false}\}$, define
\[
  \children_P(u)=\{v\mid (u,v,\sigma)\in E \wedge P(\sigma)\}.
\]
The descendant set $\desc_P(u)$ is the least set $D$ such that
$\children_P(u)\subseteq D$ and, for every $x\in D$,
$\children_P(x)\subseteq D$.  Thus $\desc_P(u)$ is the rooted reachability
set induced by edges whose states satisfy $P$.

\subsection{Histories and Budgets}

The graph captures branch structure.  The history captures ordered trace
material.  A history is a sequence
\[
  \History = \langle h_1,h_2,\ldots,h_n\rangle,
\]
where each item $h_i=(id_i,payload_i)$ has a trace identifier and a string
payload.  The payload may encode a structured event, a rendered state, or a
compact textual summary.  The model does not require a particular payload
syntax.

\begin{definition}[Budget policy]
A budget policy is a pair $(m,B)$ where $m\in\{\bytes,\tokens\}$ is a
measurement mode and $B\in\mathbb{N}$ is a nonnegative limit.  The induced
cost of a payload $x$ is either $\bytes(x)$ or $\tokens(x)$.
\end{definition}

The prototype exposes both byte and token budgets.  For fast online
estimation, it also supports the approximation
$\widehat{\tokens}(x)=\lceil |x|/4\rceil$, matching the common engineering
rule that one token is roughly four UTF-8 bytes for English-like material.
The model, however, allows an exact tokenizer to be used for measurement.

\subsection{Compaction}

Compaction maps a long history to a replacement history.  The replacement
must keep the recent suffix explicit while recording the discarded prefix in
a summary item.  Let $S$ be a summary payload.  Given budget $B$, compaction
chooses the longest suffix
\[
  \langle h_j,\ldots,h_n\rangle
\]
whose payload cost is at most $B$, truncating $h_j$ if necessary.  The
replacement is
\[
  \langle (0,S), h'_j,\ldots,h_n\rangle,
\]
where $h'_j$ is either $h_j$ or a middle-truncated version of $h_j$.  The
summary item is deliberately outside the suffix budget in our default
policy.  This choice makes the suffix budget stable: users of the structure
can reserve a fixed amount of recent material independent of the summary's
length.  An alternate policy may charge the summary against the same budget;
the algorithms are unchanged after subtracting the summary cost from $B$.

\begin{definition}[Boundary-safe truncation]
A truncation function is boundary-safe if it never splits a UTF-8 character
and if the output identifies that material was omitted.
\end{definition}

Boundary safety matters because a trace view is often parsed or displayed
after compaction.  A broken byte sequence is not a valid replacement item.
The ancillary implementation uses middle truncation with an explicit
omission marker.

\subsection{Operations}

The \BDTS{} interface supports the following operations.

\begin{itemize}[leftmargin=2em]
\item $\mathsf{upsert}(u,v,\sigma)$: insert or move the current edge for
child $v$ so that its parent is $u$ and its state is $\sigma$.
\item $\mathsf{setState}(v,\sigma)$: update the state of the current edge
whose child is $v$.
\item $\mathsf{children}(u,P)$: list the children of $u$ whose states
satisfy $P$.
\item $\mathsf{descendants}(u,P)$: list descendants of $u$ in breadth-first
order using only edges satisfying $P$.
\item $\mathsf{append}(h)$: append a trace item to the history.
\item $\mathsf{compact}(B,S)$: produce a summary-plus-suffix replacement
history under budget $B$ and summary $S$.
\end{itemize}

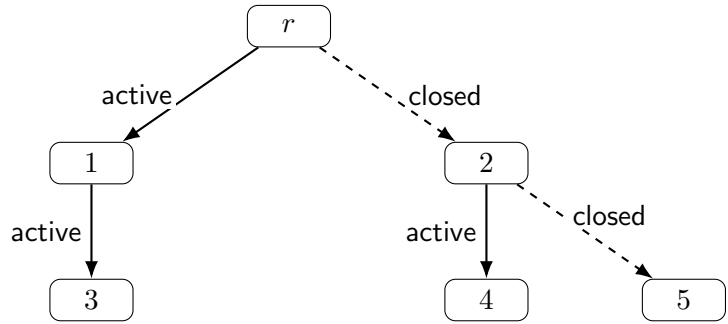
\begin{figure}[H]
\centering
\begin{tikzpicture}[
  node distance=1.25cm and 1.5cm,
  trace/.style={draw,rounded corners,minimum width=1.1cm,minimum height=.55cm},
  edge/.style={-Latex,thick},
  edgelabel/.style={fill=white,inner sep=1.4pt},
  active/.style={edge},
  closed/.style={edge,dashed}
]
\node[trace] (r) {$r$};
\node[trace,below left=of r] (a) {$1$};
\node[trace,below right=of r] (b) {$2$};
\node[trace,below=of a] (c) {$3$};
\node[trace,below=of b] (d) {$4$};
\node[trace,right=of d] (e) {$5$};
\draw[active] (r) -- node[edgelabel,left,xshift=-5pt] {$\Active$} (a);
\draw[closed] (r) -- node[edgelabel,right,xshift=6pt] {$\Closed$} (b);
\draw[active] (a) -- node[edgelabel,left,xshift=-2pt] {$\Active$} (c);
\draw[active] (b) -- node[edgelabel,left,xshift=-2pt] {$\Active$} (d);
\draw[closed] (b) -- node[edgelabel,above,pos=.62,xshift=3pt,yshift=5pt] {$\Closed$} (e);
\end{tikzpicture}
\caption{A status-filtered trace graph.  With predicate
$P(\sigma)\equiv(\sigma=\Active)$, the descendants of $r$ are $\{1,3\}$,
while the descendants of $2$ are $\{4\}$.}
\label{fig:trace-graph}
\end{figure}

\subsection{Problem Variants}

The default problem charges only the retained suffix against the user-facing
budget.  This is the most convenient policy when the summary has its own
reserved allowance.  Three variants are useful in other settings.

\emph{Charged-summary compaction} charges the summary item against the same
budget as the suffix.  If the summary cost is $s$, Algorithm~\ref{alg:compact}
is run with budget $\max(0,B-s)$ and the summary is prepended only if
$s\le B$.  If $s>B$, the summary itself must be truncated.  The complexity
is unchanged, but the retained suffix may be empty even for a positive
budget.

\emph{Lossless-backed compaction} stores the discarded prefix in a cold
archive and places a reference to that archive in the summary payload.  The
replacement history remains bounded, while exact replay is still possible
by consulting the archive.  The live data structure does not need to know
the archive representation; it needs only a stable reference.

\emph{Predicate-indexed compaction} applies different budgets to different
classes of trace items.  For example, structural items can be given a larger
retention weight than verbose payload items.  This variant replaces the
single cost function by $cost(h_i,\pi_i)$, where $\pi_i$ is an item class.
The backward scan remains the same, but the interpretation of maximality is
with respect to class-weighted cost.

The paper analyzes the default variant because it is the smallest complete
problem.  The variants show that \BDTS{} is a family of interfaces rather
than one hard-coded policy.

\section{Data Structures}

\subsection{Adjacency with State Partitions}

The abstract structure stores, for each parent $u$, adjacency buckets
$A[u,\sigma]$ keyed by state.  It also stores a map $M[v]=(u,\sigma)$ from
child to current parent and state.  With balanced dictionaries, an upsert
first checks whether $v$ already has a parent, removes it from the old
bucket if necessary, inserts it into the new bucket, and updates $M$.

For constant-size $\Sigma$, the space is $O(|V|+|E|)$.  If buckets are
stored as sorted arrays, direct listing is cache friendly and simple.  If
buckets are stored as balanced sets, movement and state update have
$O(\log d)$ worst-case cost, where $d$ is the affected parent's degree.
The ancillary crate uses sorted vectors for clarity and deterministic
output, while the analysis states the standard dictionary version.

\subsection{History Storage}

The history is append-only between compactions.  Appending an item is
$O(1)$ amortized with a growable array.  A page query can be supported by
storing offsets at fixed intervals; the artifact exposes the simpler
sequence interface, while the manuscript's model permits page boundaries.
The important invariant is that the history order is independent of graph
adjacency order.  The graph answers structural questions; the history
answers temporal and representation questions.

\subsection{Consistency Between Graph and History}

The graph and history are intentionally decoupled, but a full trace system
usually wants them to satisfy a weak consistency relation.  If a history
item refers to identifier $v$, then either $v$ is present in the graph or
the item is explicitly marked as external to the graph.  This condition is
weaker than requiring every vertex to have a history item.  Some vertices
may be structural branch points with no payload; some payloads may describe
several vertices at once.  The weak condition is enough to prevent a
compacted view from containing an unexplained identifier.

\begin{definition}[Trace-reference consistency]
A pair $(G,\History)$ is trace-reference consistent if every identifier
appearing in a non-summary history item is either a vertex of $G$ or belongs
to a declared external identifier namespace.
\end{definition}

Trace-reference consistency is preserved by edge upserts, state updates,
and descendant queries because they do not modify history payloads.  It is
preserved by append if the appended item satisfies the condition.  It is
preserved by compaction because the retained suffix is copied from the
original history and the summary item is allowed to describe a prefix rather
than refer to one current vertex.  This small invariant is helpful in
implementations that expose both graph queries and history pages through the
same interface.

\subsection{Pagination}

Long histories are rarely consumed in one scan.  A page interface returns a
slice of items plus a cursor.  For append-only histories, a cursor can be an
integer offset, a byte offset in a serialized log, or a pair consisting of a
turn identifier and an item offset.  The asymptotic behavior depends on the
chosen storage layer.  If the history is an in-memory vector, page retrieval
is $O(p)$ for page size $p$.  If the history is a line-oriented file,
retrieval from an arbitrary offset is $O(p)$ after the offset is known, but
constructing the offset may require an auxiliary sparse index.

\begin{algorithm}[htbp]
\caption{Page an append-only history}
\label{alg:page}
\begin{algorithmic}[1]
\Procedure{Page}{$\History,c,p$}
  \State $i\gets$ cursor offset represented by $c$
  \State $R\gets \langle h_i,h_{i+1},\ldots,h_{\min(i+p-1,n)}\rangle$
  \If{$i+p\le n$}
    \State $c'\gets$ cursor for offset $i+p$
  \Else
    \State $c'\gets \bot$
  \EndIf
  \State \Return $(R,c')$
\EndProcedure
\end{algorithmic}
\end{algorithm}

Pagination interacts cleanly with compaction if each compacted replacement
is treated as a new history epoch.  Cursors from old epochs can either be
rejected or mapped through an epoch table.  The ancillary artifact keeps the
core sequence interface small, but the model includes pagination because it
is the natural way to expose long traces.

\subsection{Budget Accounting}

Budget accounting is separated from the history.  A budget policy supplies
\texttt{byte\_budget()} and \texttt{token\_budget()} conversions.  Exact
tokenization can be expensive if applied repeatedly to large payloads, so
the framework permits a bounded cost cache.  The cache is not part of the
semantic state: dropping cached estimates may slow a later operation but
does not change the valid output.

\begin{proposition}[Cache noninterference]
Let $C$ be any finite cost cache whose entries are recomputable from payload
strings and a fixed budget policy.  Evicting any subset of $C$ cannot change
the graph, history, or compacted replacement selected by exact accounting.
\end{proposition}

\begin{proof}
The semantic operations depend on graph edges, history order, payload
strings, and the measurement function.  A cache entry stores only the result
of applying the measurement function to a payload.  If an entry is absent,
the same value is recomputed.  Thus cache contents affect only the running
time of measurement, not the selected suffix or graph result.
\end{proof}

\subsection{Compaction Window}

A compaction window records which compacted epoch is current and optionally
stores a prefill estimate, i.e., an estimate of representation cost already
present at the start of the window.  Starting a new window increments the
ordinal and clears the estimate.  This small structure prevents an
implementation from conflating costs measured before and after replacement.

\subsection{Soft-Capped Append Logs}

A live trace may need a small durable event log even when the full history
has a separate archival representation.  The source architecture motivates a
\emph{soft-capped log}: append entries under an exclusive writer discipline,
and whenever the byte length exceeds a hard cap $M$, trim old entries until
the remaining log is at or below a soft target $\rho M$ for some
$0<\rho\le 1$.  The newest entry is always preserved, even if it alone
exceeds the soft target.  The soft target prevents trimming on every append
after the first overflow; after a trim, several later appends can occur
before the next trim.

\begin{definition}[Soft-capped log]
A soft-capped log is a sequence $L$ with hard cap $M$ and soft ratio
$\rho$.  After appending item $x$, if $|L|_{\bytes}>M$, the structure removes
oldest items until $|L|_{\bytes}\le\max(\rho M, |x|_{\bytes})$ or only $x$
remains.
\end{definition}

\begin{lemma}[Newest-entry preservation]
Soft-cap enforcement never removes the item most recently appended.
\end{lemma}

\begin{proof}
The enforcement loop removes only from the front of the deque and stops when
one item remains.  Since the newest item is appended at the back, it cannot
be removed unless all older entries have already been removed, at which
point the loop stops.
\end{proof}

Soft-capped logs are different from summary compaction.  Compaction creates
a semantic replacement history.  Soft capping is a storage discipline for a
bounded log where retaining the newest raw entries is more important than
retaining a summary of the discarded entries.  Both structures are useful:
one supplies a bounded view, and the other supplies bounded durable recency.

\subsection{Reference-Counted Observation Registry}

Dynamic trace systems often have many consumers interested in overlapping
regions of the trace.  Registering every consumer separately with the
underlying event source is wasteful.  \BDTS{} therefore includes a
reference-counted observation registry.  Each subscriber registers keys in
one of two modes: exact, which reports only changes to that key, or
recursive, which reports changes to the key and its descendants under a
separator order.  The registry deduplicates repeated registrations,
maintains counts by key and mode, and exposes the effective mode needed by
the underlying source.

\begin{definition}[Effective observation mode]
For a key $x$, let $c_E(x)$ and $c_R(x)$ be the exact and recursive
registration counts.  The effective mode is recursive if $c_R(x)>0$, exact
if $c_R(x)=0$ and $c_E(x)>0$, and absent otherwise.
\end{definition}

This rule is a small data-structural optimization.  Recursive observation
dominates exact observation for the same key, so the underlying source needs
only the strongest active mode.  Removing a registration decrements the
appropriate counter; the source is reconfigured only when the effective mode
changes.  The ancillary artifact implements this registry over strings, but
the same rule applies to any ordered namespace.

\subsection{Delta Overlays}

The trace graph records structural lineage, while the history records
temporal payloads.  A third useful object is a \emph{delta overlay}: a
temporary map that aggregates exact changes between a baseline key-value
state and a current key-value state.  The overlay stores three maps:
baseline values by key, current values by key, and origin keys for moved
current keys.  It supports add, delete, update, and move-update operations.
If an operation is not exact, the overlay can be invalidated, after which no
exact diff is reported.

Delta overlays matter because trace compaction frequently wants a short
description of what changed during a window.  Replaying every item in the
window is expensive; storing only the final key-level delta can be much
smaller.  The overlay is not a replacement for the history, since it loses
operation order, but it is a useful auxiliary summary.  It is also a good
example of the general \BDTS{} principle: keep the exact active structure
small, and make lossy or order-insensitive summaries explicit.

\section{Algorithms}

\subsection{Graph Update and Query}

Algorithm~\ref{alg:upsert} shows edge upsert.  The operation is intentionally
total: if the child already appears under another parent, it is moved.  This
is useful for traces that are repaired or rebased.

\begin{algorithm}[htbp]
\caption{Upsert a status-labeled trace edge}
\label{alg:upsert}
\begin{algorithmic}[1]
\Procedure{Upsert}{$u,v,\sigma$}
  \If{$M$ contains $v$}
    \State $(u_0,\sigma_0)\gets M[v]$
    \State remove $v$ from $A[u_0,\sigma_0]$
  \EndIf
  \State insert $v$ into $A[u,\sigma]$
  \State $M[v]\gets(u,\sigma)$
\EndProcedure
\end{algorithmic}
\end{algorithm}

State update is a special case: find $(u,\sigma_0)=M[v]$, remove $v$ from
$A[u,\sigma_0]$, insert it into $A[u,\sigma]$, and update $M[v]$.

For descendant enumeration, breadth-first traversal uses the predicate to
select buckets.  If $P$ accepts a set $\Sigma_P\subseteq\Sigma$, then the
children of $u$ are the concatenation of $A[u,\sigma]$ for
$\sigma\in\Sigma_P$, optionally sorted by identifier for deterministic
output.  Breadth-first order is useful in trace systems because it reports
shorter branch distance before deeper material.  Depth-first order can be
substituted with the same asymptotic bound.

\begin{algorithm}[htbp]
\caption{Budgeted suffix compaction}
\label{alg:compact}
\begin{algorithmic}[1]
\Procedure{Compact}{$\History,B,S$}
  \State $R\gets$ empty list
  \State $b\gets B$
  \For{$i=n,n-1,\ldots,1$}
    \State $c\gets\tokens(h_i.payload)$
    \If{$c\le b$}
      \State prepend $h_i$ to $R$
      \State $b\gets b-c$
    \ElsIf{$b>0$}
      \State prepend $\textsc{TruncateMiddle}(h_i,b)$ to $R$
      \State $b\gets0$
      \State \textbf{break}
    \Else
      \State \textbf{break}
    \EndIf
  \EndFor
  \State prepend summary item $(0,S)$ to $R$
  \State \Return $R$
\EndProcedure
\end{algorithmic}
\end{algorithm}

\subsection{History Compaction}

Algorithm~\ref{alg:compact} gives the default compaction procedure.  It
scans backward through the history and selects as many recent payloads as
fit under budget.  The first item that does not fit is middle-truncated if
there is any remaining budget.  The summary item is then prepended.  The
algorithm is online with respect to the suffix: it never needs to inspect
discarded prefix items after the boundary has been found.

\begin{lemma}[Suffix maximality]
For a fixed nonnegative budget $B$ and nonnegative item costs, the suffix
selected by Algorithm~\ref{alg:compact} is the longest suffix whose total
cost is at most $B$, up to truncation of the boundary item.
\end{lemma}

\begin{proof}
The algorithm scans from the last item to the first, subtracting cost while
the next whole item fits.  Suppose it stops before item $h_j$ because
$cost(h_j)>b$, where $b$ is the remaining budget.  Any longer untruncated
suffix must include $h_j$ plus all already selected items, whose total cost
would exceed $B$.  If $b>0$, a boundary-safe prefix/suffix fragment of
$h_j$ is included, which is the only additional material allowed without
exceeding the budget.
\end{proof}

\subsection{Soft-Cap Enforcement}

Algorithm~\ref{alg:softcap} gives the soft-cap procedure.  The target
$\max(\rho M, |x|)$ makes a deliberate exception for very large newest
entries.  If the newest entry is larger than the soft cap, the structure
keeps that entry and discards older entries.  If the newest entry is small,
the structure trims to the soft target and thereby creates slack for future
appends.

\begin{algorithm}[htbp]
\caption{Enforce a soft byte cap after append}
\label{alg:softcap}
\begin{algorithmic}[1]
\Procedure{EnforceSoftCap}{$L,M,\rho,x$}
  \If{$|L|_{\bytes}\le M$}
    \State \Return $L$
  \EndIf
  \State $T\gets \max(\lfloor \rho M\rfloor, |x|_{\bytes})$
  \While{$|L|_{\bytes}>T$ and $L$ has more than one item}
    \State remove the oldest item from $L$
  \EndWhile
  \State \Return $L$
\EndProcedure
\end{algorithmic}
\end{algorithm}

\begin{proposition}[Amortized trimming]
Assume each appended item has byte length at most $\Delta$ and
$\rho<1$.  After a trim to target at most $\rho M$, at least
$\lfloor(1-\rho)M/\Delta\rfloor$ further appends are required before another
overflow is forced.
\end{proposition}

\begin{proof}
Immediately after trimming, the log has length at most $\rho M$ unless the
newest entry alone is larger than that value.  Under the stated bounded-item
assumption, each subsequent append increases the length by at most $\Delta$.
The log cannot exceed $M$ until the accumulated increase is greater than
$(1-\rho)M$, giving the stated lower bound.
\end{proof}

\subsection{Observation Registration}

Algorithm~\ref{alg:registry} summarizes registration in the observation
registry.  The input key list is sorted and deduplicated before counters are
updated.  This makes repeated registration of the same key idempotent for a
fixed subscriber.

\begin{algorithm}[htbp]
\caption{Register observed keys}
\label{alg:registry}
\begin{algorithmic}[1]
\Procedure{Register}{$s,K$}
  \State sort and deduplicate $K$
  \For{each $(key,mode)\in K$}
    \If{$(key,mode)$ is not already registered by $s$}
      \State add $(key,mode)$ to subscriber $s$
      \State increment counter $C[key,mode]$
      \State recompute the effective mode of $key$
    \EndIf
  \EndFor
\EndProcedure
\end{algorithmic}
\end{algorithm}

Projection of notifications is a predicate query over registered keys.
For exact mode, key $x$ matches changed key $y$ only when $x=y$.  For
recursive mode, $x$ matches $y$ when $y=x$ or when $y$ lies below $x$ in the
separator namespace.  With a trie over registered keys, projection takes
$O(|y|+z)$ time, where $z$ is the number of subscribers reported.  The
artifact uses maps for clarity; a trie is the natural asymptotic
improvement for very large subscriber sets.

\subsection{Delta Overlay Operations}

The delta overlay maintains three maps.  On update, the old value is stored
in the baseline map if the key has not already been seen, and the new value
is stored in the current map.  On delete, the current value is removed; if
the key was not already current, the supplied old value is recorded in the
baseline map.  On move-update, the source key is associated with a
destination key in the origin map.  A rename pair is reported only if the
origin key exists in the baseline, the destination key exists in the current
map, and the origin key no longer exists in the current map.

\begin{lemma}[Overlay exactness]
If all overlay operations are exact and the overlay is not invalidated, the
set of changed keys reported by the overlay is exactly the symmetric
difference between baseline and current key-value maps, plus keys whose
values differ.
\end{lemma}

\begin{proof}
The overlay inserts a baseline value the first time a key is modified and
updates the current map after each exact operation.  Keys that are never
modified are absent from both maps and therefore not reported.  A modified
key is reported precisely when its optional baseline value differs from its
optional current value.  This condition captures additions, deletions, and
value changes.
\end{proof}

\subsection{Truncation}

Middle truncation keeps both the beginning and end of a boundary payload.
This is preferable for trace items whose prefix contains an identifier or
operation name and whose suffix contains a final value.  The omission marker
is charged to the boundary item.  If a strict external budget must include
the marker, the implementation can reserve marker cost before splitting the
payload.  The ancillary prototype uses a practical approximation: split the
byte budget between left and right, align both boundaries to UTF-8 character
boundaries, and insert a marker stating the number of omitted characters.

\subsection{Cost Caching}

Exact tokenizer calls may dominate compaction when payloads are large and
repeatedly measured.  A bounded least-recently-used cache is therefore a
natural auxiliary structure.  The cache key can be a hash of the payload and
the budget mode.  Since cost values are deterministic, cache eviction is
safe by Proposition 1.  Classical paging analysis explains why recency is a
reasonable default when access locality is present \cite{sleator1985amortized};
adaptive replacement is a possible extension when scans and repeated hits
interleave \cite{megiddo2003arc}.

\section{Complexity Analysis}

Let $n=|V|$, $m=|E|$, $d(u)$ be the degree of parent $u$, and
$m_P(u)$ be the number of edges reachable from $u$ under predicate $P$.
Let $L$ be the number of history items and let $\ell_i$ be the measured
cost of item $i$.  Let $q$ be the number of items retained by compaction.

\begin{theorem}[Graph operation bounds]
Using balanced adjacency buckets and a child-to-parent map, \BDTS{} supports
edge upsert and edge-state update in $O(\log d)$ time for the affected
parent degree $d$, direct child listing in $O(k+\log d)$ time for $k$
reported children, and state-filtered descendant enumeration in
$O(m_P(u)+1)$ time after the first reported vertex.
\end{theorem}

\begin{proof}
The child-to-parent map locates the previous edge for a child in expected
$O(1)$ time with hashing or $O(\log n)$ with a balanced map.  Removing from
and inserting into a balanced adjacency bucket costs $O(\log d)$.  Listing
children under a constant number of accepted states scans exactly the
reported bucket entries, plus lookup overhead.  Breadth-first descendant
enumeration places each reported vertex in the queue once and scans exactly
the accepted outgoing edges of reported vertices.  Thus the traversal is
linear in the reachable filtered subgraph.
\end{proof}

\begin{theorem}[Compaction bounds]
With precomputed or cached item costs, Algorithm~\ref{alg:compact} runs in
$O(q+1)$ time and $O(q+1)$ output space.  Without cached costs, it runs in
$O(q+1+\sum_{i=j}^{n} |h_i|)$ time, where $h_j$ is the boundary item.
\end{theorem}

\begin{proof}
The algorithm scans retained suffix items and at most one boundary item.
With cached costs, each inspected cost is $O(1)$.  The output consists of
the summary plus retained items and possibly one truncated boundary item.
Without cached costs, the tokenizer or byte counter must inspect each
selected payload and the boundary payload, giving the stated additional
measurement term.
\end{proof}

\begin{proposition}[Space bound]
The combined structure uses $O(n+m+L+C)$ space, where $C$ is the capacity of
the auxiliary cost cache.
\end{proposition}

\begin{proof}
The graph stores one current-parent map entry per child and one adjacency
entry per edge.  The history stores $L$ item records plus payload storage.
The compaction window is constant size.  The cost cache is explicitly
bounded by $C$ entries.
\end{proof}

\begin{table}[H]
\centering
\caption{Asymptotic behavior of the abstract \BDTS{} interface.}
\label{tab:complexity}
\begin{tabular}{lll}
\toprule
Operation & Time & Additional space \\
\midrule
Edge upsert & $O(\log d)$ & $O(1)$ \\
State update & $O(\log d)$ & $O(1)$ \\
Direct child listing & $O(k+\log d)$ & $O(k)$ output \\
Filtered descendants & $O(m_P(u)+1)$ & $O(w)$ queue \\
History append & $O(1)$ amortized & $O(1)$ plus payload \\
History page & $O(p)$ after cursor resolution & $O(p)$ output \\
Budgeted compaction & $O(q+1)$ with cached costs & $O(q+1)$ output \\
Soft-cap trim & $O(t)$ for $t$ removed items & $O(1)$ \\
Observation register & $O(k\log k)$ map version & $O(k)$ \\
Notification projection & $O(s)$ map version, $O(|y|+z)$ trie version & $O(z)$ output \\
Delta overlay update & $O(\log r)$ with ordered maps & $O(1)$ new records \\
Cost-cache lookup & $O(1)$ expected & bounded by $C$ \\
\bottomrule
\end{tabular}
\end{table}

\subsection{Restricted Reachability Versus General Dynamic Reachability}

The current-parent invariant is the reason \BDTS{} avoids the hardest cases
of general dynamic reachability.  General directed reachability under
arbitrary insertions and deletions motivates sophisticated algebraic and
combinatorial techniques \cite{sankowski2004dynamic,roditty2008dynamic}.
\BDTS{} instead stores a rooted trace relation whose updates are local moves
or state changes.  Descendant enumeration remains output-sensitive because
the operation asks for the reachable set itself, not merely a Boolean
answer.  This is appropriate for trace maintenance: if a query returns
10,000 descendants, writing the answer already costs $\Omega(10,000)$.

\subsection{Budget Lower Bound}

No compaction scheme that returns an explicit suffix can avoid reading the
boundary region in the worst case.  Consider a history whose last $q$ items
all fit exactly and whose $(q+1)$st item exceeds the remaining budget by one
unit.  Any algorithm that claims to return the longest suffix must
distinguish this instance from one in which the $(q+1)$st item fits.  That
distinction requires knowing the cost of the boundary item.  Thus the
cached-cost $O(q+1)$ bound is optimal up to dictionary overhead for exact
suffix maximality.

\section{Implementation}

The ancillary artifact contains a Rust crate named
\texttt{budgeted-trace-structures}.  Its public types mirror the mathematical
objects grouped in Table~\ref{tab:artifact}: graph reachability, bounded
history, compaction windows, budget policies, soft-capped logs, observation
registries, delta overlays, cost caches, and benchmark drivers.  The
implementation uses deterministic sorted vectors for adjacency lists in
order to keep outputs stable in tests and benchmark files.  This choice
keeps the artifact compact and readable.  Replacing each adjacency vector by
a balanced set gives the asymptotic update bounds in
Table~\ref{tab:complexity} without changing the interface.

\begin{table}[H]
\centering
\caption{Artifact components and their role.}
\label{tab:artifact}
\begin{tabular}{lll}
\toprule
Component & Public object & Purpose \\
\midrule
Trace graph & \texttt{TraceGraph} & status-filtered child and descendant queries \\
History & \texttt{BudgetedHistory} & append-only payload sequence \\
Budget policy & \texttt{BudgetPolicy} & byte/token accounting interface \\
Window state & \texttt{CompactionWindow} & compacted epoch and prefill estimate \\
Auxiliary cache & \texttt{BoundedCostCache} & bounded recomputable cost estimates \\
Soft log & \texttt{SoftCappedLog} & hard/soft bounded durable recency \\
Observation registry & \texttt{ObservationRegistry} & reference-counted exact/recursive keys \\
Delta overlay & \texttt{DeltaOverlay} & exact key-level change aggregation \\
Benchmark & \texttt{TraceBench} & synthetic measurement driver \\
\bottomrule
\end{tabular}
\end{table}

\subsection{Engineering Choices}

The prototype makes three conservative choices.  First, graph identifiers
are integer trace identifiers, avoiding allocation-heavy keys inside the
hot path.  Second, cost accounting is explicit rather than hidden in the
history.  This makes it possible to run the same trace under byte budgets,
approximate token budgets, or exact tokenizer measurements.  Third, the
benchmark driver emits JSON and CSV.  JSON preserves structured metadata;
CSV is convenient for plots and tables.

The Python measurement script in the artifact loads
\texttt{distilbert/distilgpt2}, \texttt{gpt2}, and
\texttt{facebook/opt-125m}.  The first two use the GPT-2 byte-level BPE
family \cite{radford2019language,hfDistilgpt2}; the third is from the OPT
family \cite{zhang2022opt}.  The script uses these models only as concrete
tokenizer and forward computation targets.  The measured quantity is
representation cost, not output quality.

\section{Experiments}

\subsection{Questions}

The experiments answer three questions.

\begin{enumerate}[leftmargin=2em]
\item How quickly does the prototype build and query a status-filtered trace
graph with tens of thousands of vertices?
\item How much can a budgeted summary-plus-suffix representation reduce a
long synthetic history?
\item Does the reduction remain visible under an actual byte-level BPE
tokenizer and a small forward computation?
\end{enumerate}

\subsection{Synthetic Trace Matrix}

The synthetic benchmark matrix constructs rooted traces with 10,000, 20,000,
and 40,000 vertices.  The three workloads vary branching factor, state
period, payload length, and budget.  Each workload assigns closed state to
children whose identifiers are divisible by a workload-specific period and
active state to all other edges.  The history contains one payload per
vertex, and compaction retains a summary plus the largest suffix fitting the
approximate token budget.

\begin{table}[H]
\centering
\caption{Measured Rust benchmark matrix.  Times are milliseconds.}
\label{tab:synthetic}
\resizebox{\textwidth}{!}{%
\begin{tabular}{lrrrrrrrr}
\toprule
Workload & Vertices & Edges & Active desc. & All desc. & Build & Active query & Full query & Compact \\
\midrule
balanced\_10k & 10,000 & 9,999 & 3,624 & 9,999 & 0.582 & 0.107 & 0.243 & 0.0056 \\
wide\_20k & 20,000 & 19,999 & 7,854 & 19,999 & 1.169 & 0.255 & 0.589 & 0.0091 \\
deep\_40k & 40,000 & 39,999 & 10,954 & 39,999 & 2.722 & 0.689 & 1.415 & 0.0390 \\
\bottomrule
\end{tabular}
}
\end{table}

\begin{table}[H]
\centering
\caption{Budget and auxiliary-structure outcomes from the same Rust matrix.}
\label{tab:budget-matrix}
\resizebox{\textwidth}{!}{%
\begin{tabular}{lrrrrrr}
\toprule
Workload & Original tok. & Compact tok. & Ratio & Soft-log entries & Soft-log bytes & Registry time \\
\midrule
balanced\_10k & 350,000 & 1,048 & 0.002994 & 110 & 15,400 & 0.00158 ms \\
wide\_20k & 1,030,000 & 2,072 & 0.002012 & 116 & 23,780 & 0.00238 ms \\
deep\_40k & 2,710,000 & 4,120 & 0.001520 & 100 & 26,900 & 0.00454 ms \\
\bottomrule
\end{tabular}
}
\end{table}

Tables~\ref{tab:synthetic} and \ref{tab:budget-matrix} show three effects.
First, graph construction and traversal scale with the output size, as
predicted by the analysis.  The 40,000-vertex workload enumerates 39,999
descendants in 1.415 ms and the active filtered subset in 0.689 ms.  Second,
summary-plus-suffix compaction reduces histories by two to three orders of
magnitude under these budgets.  Third, the soft-capped log and observation
registry remain small auxiliary structures: the log retains roughly one
hundred recent entries, and registry projection is below 0.005 ms in all
three workloads.

\subsection{Tokenizer and Forward Matrix}

The model-budget measurement builds two strings for each model.  The raw
string contains 160 event lines with 112-byte payloads.  The compact string
contains a summary line and 20 retained event lines.  Both strings are
tokenized.  The compact string is also passed through a forward computation
over a 256-token window and a deterministic eight-token generation step over
a 128-token window.  The raw string is used only for token counting when it
exceeds a model's nominal context length.

\begin{table}[H]
\centering
\caption{Measured tokenizer and forward matrix.  Times are milliseconds.}
\label{tab:hf}
\resizebox{\textwidth}{!}{%
\begin{tabular}{lrrrrrrr}
\toprule
Model & Context & Raw tok. & Compact tok. & Ratio & Load & Forward & Generate \\
\midrule
\texttt{distilbert/distilgpt2} & 1,024 & 3,359 & 432 & 0.12861 & 1,728.0 & 66.8 & 50.9 \\
\texttt{gpt2} & 1,024 & 3,359 & 432 & 0.12861 & 1,904.5 & 113.6 & 85.6 \\
\texttt{facebook/opt-125m} & 2,048 & 3,360 & 433 & 0.12887 & 1,975.3 & 70.9 & 87.7 \\
\bottomrule
\end{tabular}
}
\end{table}

The token ratio in Table~\ref{tab:hf} is stable across two GPT-2-family
tokenizers and the OPT tokenizer: the compact representation uses about
12.9\% of the raw tokens.  The load time varies substantially because the
measurement includes local model initialization and first-use cache effects.
The forward and generation columns document that the compact representation
is accepted by each concrete computation target.  The OPT measurement is
included to avoid relying on a single tokenizer family and to test a larger
context limit.

\subsection{Reproducibility}

The artifact records the exact commands in its README.  The aggregate result
files are:
\begin{itemize}[leftmargin=2em]
\item \texttt{anc/results/tracebench\_matrix.json} and
\texttt{anc/results/tracebench\_matrix.csv};
\item \texttt{anc/results/model\_matrix.json} and
\texttt{anc/results/model\_matrix.csv}.
\end{itemize}
The manuscript tables were copied from those generated files.  The arXiv
source bundle places the ancillary files under the root-level \texttt{anc/}
directory, following the arXiv ancillary-file convention for code and data
material.

\section{Discussion}

\subsection{Why Summary Plus Suffix}

Pure compression and pure suffix retention solve different problems.
Compression can preserve many bytes but may obscure the most recent items.
Suffix retention preserves recency but forgets all older material.  \BDTS
uses a hybrid: a summary item stands for the prefix, while the suffix keeps
recent trace items explicit.  This is a data-structural compromise.  It does
not require the summary to be perfect; it only requires the replacement
history to make the boundary visible and enforce the budget.

\subsection{State-Filtered Reachability}

The edge-state predicate is a small addition with large practical effect.
Many trace systems need to distinguish currently active branches from
closed, interrupted, or archived branches.  Encoding that status in the
edge relation allows queries such as ``give the active descendants of this
root'' without scanning unrelated material.  The model also supports future
states by treating states as strings or small enums.  If the number of
states grows, buckets can be indexed by a dictionary rather than an array.

\subsection{Why Reference Counts Belong in the Structure}

Observation registration can be implemented outside the trace structure, but
placing it beside the trace has two advantages.  First, the registry shares
the same namespace as the graph and history, so recursive observation of a
subtree can be interpreted consistently with descendant queries.  Second,
reference counts make reconfiguration proportional to changes in effective
mode rather than proportional to the number of subscribers.  If one hundred
subscribers observe the same recursive key, the underlying source needs one
recursive registration, not one hundred.

The exact/recursive dominance rule is intentionally simple.  Recursive
mode dominates exact mode for the same key because every exact event is also
reported by recursive observation.  For different keys, no dominance is
assumed unless the namespace relation implies it.  This keeps the registry
local: it does not need to solve a global set-cover problem over observed
regions.  A trie implementation can reduce projection time while preserving
the same semantics.

\subsection{Soft Caps Versus Sliding Windows}

A sliding window usually retains the most recent items whose total cost is
below a threshold.  A soft-capped log is slightly different: it permits the
log to grow up to a hard cap, then trims below that cap.  The gap between
the hard cap and the soft target creates hysteresis.  Hysteresis is useful
when trimming has a fixed overhead, for example when rewriting a line-based
file or updating a sparse offset index.  The resulting structure is still
bounded by $M$ after enforcement, but it avoids turning every later append
into a trim operation.

\subsection{Delta Overlays and Compaction Summaries}

Delta overlays are a bridge between exact histories and compact summaries.
They preserve the final key-level effect of a window without preserving the
full order of operations.  This makes them unsuitable for exact replay, but
well suited for summary construction, conflict detection, and change
headers.  A compaction routine can include a delta overlay in the summary
payload, making the summary more structured than free text while keeping
the suffix budget unchanged.

\subsection{Approximate Versus Exact Tokens}

Approximate token counts are useful inside hot paths because they are
cheap, deterministic, and conservative enough for engineering decisions.
Exact tokenizer counts are useful for final accounting and experiments.
The separation of \texttt{BudgetPolicy} from \texttt{BudgetedHistory} is
therefore important: the same history can be compacted under a fast
estimate or measured with an external tokenizer after the fact.  Subword
tokenization methods such as BPE \cite{gage1994bpe,sennrich2016bpe} and
SentencePiece \cite{kudo2018sentencepiece} make this distinction explicit.

\subsection{Relation to Succinctness}

\BDTS{} is not a succinct data structure in the strict static sense of
Jacobson-style encodings \cite{jacobson1989succinct,navarro2016compact}.
It is a dynamic representation-budget structure.  The objective is not to
approach an information-theoretic lower bound for a fixed tree, but to
preserve useful dynamic operations under an explicit representation limit.
Nevertheless, succinct techniques suggest future encodings for cold trace
segments.  A completed prefix could be stored in a compressed tree sequence
while the active suffix remains in mutable form.

\section{Related Work}

\subsection{Dynamic Graph Algorithms}
The graph component descends from classical traversal and dynamic
reachability.  Tarjan's depth-first search gives the linear-time baseline
for static graph analysis \cite{tarjan1972dfs}.  Even and Shiloach studied
online deletion in graph problems \cite{even1981online}.  Henzinger and
King developed randomized fully dynamic graph algorithms with
polylogarithmic update time for several graph properties
\cite{henzinger1999randomized}.  Dynamic transitive closure and reachability
have also been attacked through matrix and combinatorial methods
\cite{sankowski2004dynamic,roditty2008dynamic}.  \BDTS{} is more restricted:
it focuses on rooted output-sensitive enumeration under a current-parent
trace invariant.

\subsection{Static-to-Dynamic Transformations and Search}
Bentley and Saxe's decomposable searching framework shows how static
structures can be lifted into dynamic settings \cite{bentley1980decomposable}.
The \BDTS{} graph can be seen as the opposite design point: it begins with a
restricted dynamic interface and then asks which static summaries may be
attached to inactive regions.  Efficient string matching, such as the
Aho--Corasick automaton \cite{aho1975string}, is relevant when payloads
must be searched after compaction, but payload search is orthogonal to the
trace graph.

\subsection{Succinct and Compact Structures}
Succinct tree and dictionary representations \cite{jacobson1989succinct,raman2002succinct}
and broader compact-data-structure treatments \cite{navarro2016compact}
are relevant for storing old trace regions.  Their main target is static or
semi-static representation.  \BDTS{} instead keeps the active region mutable
and budgeted, but it can use succinct encodings for frozen prefixes.

\subsection{Streaming and Windowed Summaries}
Streaming algorithms maintain compact summaries of large input sequences
\cite{muthukrishnan2005streams}.  Sliding-window statistics
\cite{datar2002sliding} and sketches such as Count-Min
\cite{cormode2005countmin} show how approximate answers can be kept under
space constraints.  \BDTS{} differs by preserving an explicit suffix and a
valid replacement history.  The compacted result is not only a statistic;
it remains a sequence of trace items.

\subsection{Compression and Tokenization}
Universal compression algorithms \cite{ziv1977universal,ziv1978compression}
and dictionary methods such as LZW \cite{welch1984technique} reduce byte
representations.  BPE \cite{gage1994bpe} and its use in subword processing
\cite{sennrich2016bpe,kudo2018sentencepiece} motivate token-budget
accounting.  Distilled transformer models
\cite{sanh2019distilbert,hinton2015distilling}, GPT-2-family models
\cite{radford2019language}, OPT \cite{zhang2022opt}, and the
\texttt{distilbert/distilgpt2} card \cite{hfDistilgpt2} provide concrete
tokenizer targets for the experiments.

\section{Limitations}

The current formulation uses a single current parent per trace vertex.
This is appropriate for rooted trace branching, but it does not express a
general directed graph with arbitrary joins.  A future version could allow
multiple parents by replacing the child-to-parent map with a child-to-edge
set, at the cost of more expensive movement and duplicate-control logic.

The compaction summary is treated as an input to the data structure rather
than derived by a fixed semantic algorithm.  This is deliberate: summary
quality is application-dependent, while the data-structure problem is to
place, budget, and preserve the summary-plus-suffix replacement.  For tasks
requiring exact replay, compaction should be paired with a lossless archive
of discarded items.

The benchmark is intentionally small enough to be reproducible on a local
workstation.  It does not claim to characterize all hardware, all
tokenizers, or all trace distributions.  Its purpose is to verify that the
artifact performs the operations described in the analysis and to provide a
concrete measured table rather than an illustrative estimate.

Finally, approximate token accounting can undercount or overcount for
languages and payload formats far from the English-like byte distribution
implicit in the four-byte rule.  Exact tokenizer measurements should be used
when a hard external limit must be satisfied.

\section{Conclusion}

\BDTS{} frames token-efficient sequential computation as a data-structures
problem.  The structure maintains a status-filtered rooted trace graph, an
append-only history, a soft-capped log, a reference-counted observation
registry, a delta overlay, a byte/token budget interface, and a compaction
window that replaces old material by a summary plus explicit suffix.  The
resulting operations have simple output-sensitive bounds and a compact
implementation.  The measured artifact demonstrates that 10,000--40,000
vertex traces can be built, queried, observed, paged, and compacted quickly,
and that the representation reduction is visible across multiple concrete
tokenizer and forward-computation targets.  The broader lesson is that
representation budgets should be first-class in dynamic trace structures:
when a computation must carry its history forward, the history needs the
same algorithmic care as the computation itself.

\appendix

\section{Proof Details}

\subsection{Deterministic Breadth-First Order}

If each adjacency bucket is sorted by child identifier, breadth-first
enumeration is deterministic.  Let $Q$ be the queue after processing $t$
vertices.  The next vertex is the earliest enqueued unprocessed child among
all vertices at the current frontier distance.  Sorting buckets fixes the
within-parent child order, and queue discipline fixes the between-parent
order.  By induction on frontier distance, the output is uniquely
determined by the graph and predicate.

\subsection{Replacement Validity}

A replacement history is valid if it is a finite sequence of trace items
whose payloads are valid strings and whose first item is a summary item.
Algorithm~\ref{alg:compact} returns a valid replacement history because it
constructs a finite list, copies valid payloads from the original history,
and applies boundary-safe truncation to at most one payload.  The summary
payload is supplied as a valid string.  Thus every output payload is valid.

\subsection{Budget Monotonicity}

Let $R(B)$ be the retained suffix selected by Algorithm~\ref{alg:compact}
under budget $B$, ignoring the summary item.  If $B_1\le B_2$, then
$R(B_1)$ is a suffix of $R(B_2)$, except that the boundary item under
$B_1$ may be a shorter truncation of the corresponding boundary item under
$B_2$.  This follows because the backward scan under $B_2$ has at least as
much remaining budget at every step as the scan under $B_1$.

\section{Additional Implementation Notes}

\subsection{Why the Artifact Uses Sorted Vectors}

The artifact favors transparency.  A sorted vector adjacency list is easy
to inspect, serializes naturally, and gives deterministic output without an
additional ordering layer.  For the benchmark size used in the paper, the
constant factors are small.  For workloads with frequent movement under
high-degree parents, a balanced-set bucket is the direct replacement.

\subsection{CSV and JSON Outputs}

The benchmarks write both structured JSON and CSV files.  JSON is used by
automated checks because field names reduce ambiguity.  CSV is used for
manuscript tables because it can be imported directly into spreadsheet or
plotting tools.  The matrix drivers generate one per-workload JSON file and
one aggregate matrix file.

\section{Extended Example: End-to-End Maintenance}

This section gives a complete operational example that combines the
components analyzed above.  Consider a sequential computation that begins
at root $r=0$ and emits a new trace vertex whenever it reaches a durable
state.  Each vertex stores only an identifier in the graph; the larger
payload is appended to the history.  Suppose vertices $1,2,3$ branch from
the root, vertex $4$ branches from $1$, and vertex $5$ branches from $4$.
If vertex $2$ is later closed, the edge $(0,2)$ remains in the graph but its
state changes from $\Active$ to $\Closed$.  A query for active descendants
of $0$ now reports $1,3,4,5$ and omits $2$.  A query without the state
predicate still reports all five descendants.  The update has not rewritten
history items; it has changed only the structural index.

At the same time, the history receives payloads
$h_1,\ldots,h_5$.  A page request with page size two returns
$\langle h_1,h_2\rangle$ and a cursor for the third item.  A later page
request resumes from that cursor.  Because the history is append-only within
the epoch, existing cursors remain meaningful until compaction creates a new
epoch.  This separation is useful: graph queries can be answered in
structural order, while pages are answered in temporal order.

Now suppose a client observes the key \texttt{root} recursively and another
client observes \texttt{root/branch/4} exactly.  The registry stores one
recursive count for \texttt{root} and one exact count for
\texttt{root/branch/4}.  A change at \texttt{root/branch/4/value} is
projected to the recursive observer but not the exact observer.  A change
at \texttt{root/branch/4} is projected to both.  If several clients observe
the same recursive key, the registry increments a counter rather than
duplicating the effective observation mode.  Removing a client decrements
the counter and changes the effective mode only when the counter reaches
zero.

During this window, a delta overlay records the exact key-level effect of
payload updates.  If key $a$ changes from $x$ to $y$, the overlay records
baseline value $x$ and current value $y$.  If $a$ is then moved to $b$ and
updated to $z$, the overlay records the rename relation $(a,b)$ and the
current value at $b$.  The ordered history still contains both operations;
the overlay records the final window effect.  This distinction allows a
summary to contain a small exact change header while the retained suffix
keeps the most recent ordered items.

Finally, suppose the representation budget is $B=512$ tokens and the
history has grown to many thousands of items.  Compaction scans backward
from the newest item, inserts whole items while their costs fit, and
truncates one boundary item if needed.  The replacement history begins with
a summary item that can include the current epoch number, the number of
discarded items, and the delta overlay's changed keys.  The selected suffix
then follows verbatim.  The graph is not rebuilt by this operation: the
graph continues to answer structural queries, while the history epoch has
changed to a summary-plus-suffix representation.

\begin{proposition}[Composed batch cost]
Consider a batch with $a$ graph upserts, $s$ state updates, $h$ history
appends, $o$ observation registrations after deduplication, and one
compaction retaining $q$ items.  With balanced adjacency buckets, cached
payload costs, and map-based observation counters, the batch takes
\[
  O((a+s)\log d_{\max} + h + o\log o + q)
\]
time plus output size, where $d_{\max}$ is the maximum affected parent
degree in the batch.
\end{proposition}

\begin{proof}
Each upsert or state update touches a child-to-parent entry and at most two
adjacency buckets, costing $O(\log d_{\max})$ under balanced buckets.  Each
append is amortized $O(1)$.  The registration list is sorted and
deduplicated in $O(o\log o)$ time before counters are updated.  Compaction
with cached costs inspects the retained suffix and one possible boundary
item, giving $O(q)$ time.  The operations touch disjoint structural
components except for the identifier namespace, so their costs add.
\end{proof}

\begin{table}[H]
\centering
\caption{End-to-end maintenance roles in one trace epoch.}
\label{tab:end-to-end}
\begin{tabular}{lll}
\toprule
Object & Maintained information & Main query or update \\
\midrule
Trace graph & current rooted parent relation & filtered descendants \\
History & ordered payload sequence & page and compact \\
Soft log & bounded newest raw entries & append with hysteresis \\
Observation registry & counted key subscriptions & projected notification set \\
Delta overlay & baseline/current key differences & compact change header \\
Cost cache & recomputable payload costs & repeated budget accounting \\
Compaction window & epoch and prefill estimate & replacement boundary \\
\bottomrule
\end{tabular}
\end{table}

The example highlights why \BDTS{} is not just a graph, a log, or a
compressor.  The graph supplies rooted reachability; the history supplies
ordered representation; the registry and overlay supply small auxiliary
summaries; the soft log supplies bounded recent durability; and the
compaction window states where one bounded replacement view ends and the
next begins.  Each object has a narrow invariant, but the combination gives
a trace representation that can be queried, transmitted, and shortened
without changing its structural vocabulary.

\section{Artifact Checklist}

The ancillary artifact contains:
\begin{itemize}[leftmargin=2em]
\item a Rust library with the trace graph, budgeted history, compaction
window, budget policy, bounded cost cache, soft-capped log, observation
registry, and delta overlay;
\item a Rust benchmark binary and a matrix runner that emit JSON and CSV;
\item Python scripts that run tokenizer and compact forward measurements
for the three reported model targets;
\item generated result files used in Tables~\ref{tab:synthetic},
\ref{tab:budget-matrix}, and \ref{tab:hf};
\item a README with commands for reproducing the measurements.
\end{itemize}

\bibliographystyle{plainnat}
\bibliography{references}

@article{tarjan1972dfs,
  author = {Tarjan, Robert E.},
  title = {Depth-First Search and Linear Graph Algorithms},
  journal = {SIAM Journal on Computing},
  volume = {1},
  number = {2},
  pages = {146--160},
  year = {1972},
  doi = {10.1137/0201010}
}

@article{even1981online,
  author = {Even, Shimon and Shiloach, Yossi},
  title = {An On-Line Edge-Deletion Problem},
  journal = {Journal of the ACM},
  volume = {28},
  number = {1},
  pages = {1--4},
  year = {1981},
  doi = {10.1145/322234.322235}
}

@inproceedings{henzinger1999randomized,
  author = {Henzinger, Monika Rauch and King, Valerie},
  title = {Randomized Fully Dynamic Graph Algorithms with Polylogarithmic Time per Operation},
  booktitle = {Proceedings of the Thirty-First Annual ACM Symposium on Theory of Computing},
  pages = {519--527},
  year = {1999},
  doi = {10.1145/301250.301383}
}

@article{demetrescu2004dynamic,
  author = {Demetrescu, Camil and Italiano, Giuseppe F.},
  title = {A New Approach to Dynamic All Pairs Shortest Paths},
  journal = {Journal of the ACM},
  volume = {51},
  number = {6},
  pages = {968--992},
  year = {2004},
  doi = {10.1145/1039488.1039492}
}

@inproceedings{sankowski2004dynamic,
  author = {Sankowski, Piotr},
  title = {Dynamic Transitive Closure via Dynamic Matrix Inverse},
  booktitle = {45th Annual IEEE Symposium on Foundations of Computer Science},
  pages = {509--517},
  year = {2004},
  doi = {10.1109/FOCS.2004.25}
}

@article{roditty2008dynamic,
  author = {Roditty, Liam and Zwick, Uri},
  title = {Improved Dynamic Reachability Algorithms for Directed Graphs},
  journal = {SIAM Journal on Computing},
  volume = {37},
  number = {5},
  pages = {1455--1471},
  year = {2008},
  doi = {10.1137/060650271}
}

@inproceedings{jacobson1989succinct,
  author = {Jacobson, Guy},
  title = {Space-Efficient Static Trees and Graphs},
  booktitle = {30th Annual Symposium on Foundations of Computer Science},
  pages = {549--554},
  year = {1989},
  doi = {10.1109/SFCS.1989.63533}
}

@inproceedings{raman2002succinct,
  author = {Raman, Rajeev and Raman, Venkatesh and Rao, S. Srinivasa},
  title = {Succinct Indexable Dictionaries with Applications to Encoding k-ary Trees and Multisets},
  booktitle = {Proceedings of the Thirteenth Annual ACM-SIAM Symposium on Discrete Algorithms},
  pages = {233--242},
  year = {2002}
}

@book{navarro2016compact,
  author = {Navarro, Gonzalo},
  title = {Compact Data Structures: A Practical Approach},
  publisher = {Cambridge University Press},
  year = {2016},
  doi = {10.1017/CBO9781316588284}
}

@article{sleator1985amortized,
  author = {Sleator, Daniel D. and Tarjan, Robert E.},
  title = {Amortized Efficiency of List Update and Paging Rules},
  journal = {Communications of the ACM},
  volume = {28},
  number = {2},
  pages = {202--208},
  year = {1985},
  doi = {10.1145/2786.2793}
}

@inproceedings{megiddo2003arc,
  author = {Megiddo, Nimrod and Modha, Dharmendra S.},
  title = {ARC: A Self-Tuning, Low Overhead Replacement Cache},
  booktitle = {Proceedings of the 2nd USENIX Conference on File and Storage Technologies},
  pages = {115--130},
  year = {2003}
}

@book{muthukrishnan2005streams,
  author = {Muthukrishnan, S.},
  title = {Data Streams: Algorithms and Applications},
  publisher = {Now Publishers},
  year = {2005},
  doi = {10.1561/0400000002}
}

@article{datar2002sliding,
  author = {Datar, Mayur and Gionis, Aristides and Indyk, Piotr and Motwani, Rajeev},
  title = {Maintaining Stream Statistics over Sliding Windows},
  journal = {SIAM Journal on Computing},
  volume = {31},
  number = {6},
  pages = {1794--1813},
  year = {2002},
  doi = {10.1137/S0097539701398363}
}

@article{cormode2005countmin,
  author = {Cormode, Graham and Muthukrishnan, S.},
  title = {An Improved Data Stream Summary: The Count-Min Sketch and Its Applications},
  journal = {Journal of Algorithms},
  volume = {55},
  number = {1},
  pages = {58--75},
  year = {2005},
  doi = {10.1016/j.jalgor.2003.12.001}
}

@article{ziv1977universal,
  author = {Ziv, Jacob and Lempel, Abraham},
  title = {A Universal Algorithm for Sequential Data Compression},
  journal = {IEEE Transactions on Information Theory},
  volume = {23},
  number = {3},
  pages = {337--343},
  year = {1977},
  doi = {10.1109/TIT.1977.1055714}
}

@article{ziv1978compression,
  author = {Ziv, Jacob and Lempel, Abraham},
  title = {Compression of Individual Sequences via Variable-Rate Coding},
  journal = {IEEE Transactions on Information Theory},
  volume = {24},
  number = {5},
  pages = {530--536},
  year = {1978},
  doi = {10.1109/TIT.1978.1055934}
}

@article{welch1984technique,
  author = {Welch, Terry A.},
  title = {A Technique for High-Performance Data Compression},
  journal = {Computer},
  volume = {17},
  number = {6},
  pages = {8--19},
  year = {1984},
  doi = {10.1109/MC.1984.1659158}
}

@article{gage1994bpe,
  author = {Gage, Philip},
  title = {A New Algorithm for Data Compression},
  journal = {C Users Journal},
  volume = {12},
  number = {2},
  pages = {23--38},
  year = {1994}
}

@inproceedings{sennrich2016bpe,
  author = {Sennrich, Rico and Haddow, Barry and Birch, Alexandra},
  title = {Neural Machine Translation of Rare Words with Subword Units},
  booktitle = {Proceedings of the 54th Annual Meeting of the Association for Computational Linguistics},
  pages = {1715--1725},
  year = {2016},
  doi = {10.18653/v1/P16-1162}
}

@inproceedings{kudo2018sentencepiece,
  author = {Kudo, Taku and Richardson, John},
  title = {SentencePiece: A Simple and Language Independent Subword Tokenizer and Detokenizer for Neural Text Processing},
  booktitle = {Proceedings of the 2018 Conference on Empirical Methods in Natural Language Processing: System Demonstrations},
  pages = {66--71},
  year = {2018},
  doi = {10.18653/v1/D18-2012}
}

@inproceedings{sanh2019distilbert,
  author = {Sanh, Victor and Debut, Lysandre and Chaumond, Julien and Wolf, Thomas},
  title = {DistilBERT, a Distilled Version of BERT: Smaller, Faster, Cheaper and Lighter},
  booktitle = {NeurIPS Workshop on Energy Efficient Machine Learning and Cognitive Computing},
  year = {2019}
}

@inproceedings{hinton2015distilling,
  author = {Hinton, Geoffrey and Vinyals, Oriol and Dean, Jeff},
  title = {Distilling the Knowledge in a Neural Network},
  booktitle = {NIPS Deep Learning and Representation Learning Workshop},
  year = {2015}
}

@misc{radford2019language,
  author = {Radford, Alec and Wu, Jeffrey and Child, Rewon and Luan, David and Amodei, Dario and Sutskever, Ilya},
  title = {Language Models are Unsupervised Multitask Learners},
  year = {2019},
  note = {Technical report}
}

@article{zhang2022opt,
  author = {Zhang, Susan and Roller, Stephen and Goyal, Naman and Artetxe, Mikel and Chen, Moya and Chen, Shuohui and Dewan, Christopher and Diab, Mona and Li, Xian and Lin, Xi Victoria and Mihaylov, Todor and Ott, Myle and Shleifer, Sam and Shuster, Kurt and Simig, Daniel and Koura, Punit Singh and Sridhar, Anjali and Wang, Tianlu and Zettlemoyer, Luke},
  title = {OPT: Open Pre-trained Transformer Language Models},
  journal = {arXiv preprint arXiv:2205.01068},
  year = {2022},
  doi = {10.48550/arXiv.2205.01068}
}

@misc{hfDistilgpt2,
  author = {{Hugging Face}},
  title = {distilbert/distilgpt2 Model Card},
  year = {2026},
  howpublished = {\url{https://huggingface.co/distilbert/distilgpt2}},
  note = {Accessed 20 May 2026}
}

@article{aho1975string,
  author = {Aho, Alfred V. and Corasick, Margaret J.},
  title = {Efficient String Matching: An Aid to Bibliographic Search},
  journal = {Communications of the ACM},
  volume = {18},
  number = {6},
  pages = {333--340},
  year = {1975},
  doi = {10.1145/360825.360855}
}

@book{cormen2009introduction,
  author = {Cormen, Thomas H. and Leiserson, Charles E. and Rivest, Ronald L. and Stein, Clifford},
  title = {Introduction to Algorithms},
  edition = {3},
  publisher = {MIT Press},
  year = {2009}
}

@article{bentley1980decomposable,
  author = {Bentley, Jon Louis and Saxe, James B.},
  title = {Decomposable Searching Problems I: Static-to-Dynamic Transformation},
  journal = {Journal of Algorithms},
  volume = {1},
  number = {4},
  pages = {301--358},
  year = {1980},
  doi = {10.1016/0196-6774(80)90015-2}
}

\end{document}